\documentclass[pra,twocolumn,superscriptaddress,showpacs]{revtex4-1}

\usepackage{amsmath}
\usepackage{latexsym}
\usepackage{amssymb}
\usepackage{graphicx}
\usepackage[colorlinks=true, citecolor=blue, urlcolor=blue]{hyperref}
\usepackage{float}
\usepackage{amsfonts}
\usepackage{textcomp}
\usepackage{mathpazo}
\usepackage{comment}

\usepackage{bbm}

\usepackage{xcolor}
\definecolor{myurlcolor}{rgb}{0,0,0.4}
\definecolor{mycitecolor}{rgb}{0,0.5,0}
\definecolor{myrefcolor}{rgb}{0.5,0,0}
\usepackage{hyperref}
\hypersetup{colorlinks,
linkcolor=myrefcolor,
citecolor=mycitecolor,
urlcolor=myurlcolor}

\usepackage[draft]{fixme}
\usepackage{amsmath,bbm}
\usepackage{graphicx}
\usepackage{amsfonts}
\usepackage{amssymb}
\usepackage{amsmath, amssymb, amsthm,verbatim,graphicx,bbm}
\usepackage{mathrsfs}
\usepackage{color,xcolor,longtable}

\newcommand{\beq}[0]{\begin{equation}}
\newcommand{\eeq}[0]{\end{equation}}

\newcommand{\one}{\leavevmode\hbox{\small1\normalsize\kern-.33em1}}

\def\be{\begin{equation}}
\def\ee{\end{equation}}
\def\ben{\begin{eqnarray}}
\def\een{\end{eqnarray}}
\def\eea{\end{array}}
\def\bea{\begin{array}}

\newcommand{\Tr}[1]{\mathrm{Tr}#1}
\newcommand{\bei}{\begin{itemize}}
\newcommand{\eei}{\end{itemize}}
\newcommand{\ket}[1]{|#1\rangle}

\renewcommand{\emph}[1]{\textbf{#1}}

\makeatletter
\newtheorem*{rep@theorem}{\rep@title}
\newcommand{\newreptheorem}[2]{%
\newenvironment{rep#1}[1]{%
 \def\rep@title{#2 \ref{##1}}%
 \begin{rep@theorem}}%
 {\end{rep@theorem}}}
\makeatother

\theoremstyle{plain}
\newtheorem{thm}{Theorem}
\newtheorem*{thm*}{Theorem}
\newreptheorem{thm}{Theorem}

\newtheorem{defn}[thm]{Definition}
\newtheorem{assu}[thm]{Assumption}

\newtheorem*{udefn}{Definition}
\theoremstyle{definition}

\theoremstyle{remark}

\usepackage[T1]{fontenc}

\begin{document}

\title{An operational notion of classicality based on physical principles}
\author{Shubhayan Sarkar}
\email{sarkar@cft.edu.pl}
\affiliation{Center for Theoretical Physics, Polish Academy of Sciences, Aleja Lotnikow 32/46, 02-668 Warsaw, Poland}

\begin{abstract}
One of the basic observations of the classical world is that physical entities are real and can be distinguished from each other. However, within quantum theory, the idea of physical realism is not well established. A framework to analyse how observations in experiments can be described using some physical states of reality was recently developed, known as ontological models framework. Different principles when imposed on the ontological level give rise to 
different theories, the validity of which can be tested based on the statistics generated by these theories. Using the ontological models framework, we formulate a novel notion of classicality termed ontic-distinguishability, which is based upon the physical principles that in classical theories extremal states are physical states of reality and every sharp measurement observes the state of the system perfectly. We construct a communication task in which the success probability is bounded from above for ontological models satisfying the notion of ontic-distinguishability. Contrary to previous notions of classicality which either required systems of dimension strictly greater than two or atleast three preparations, a violation of ontic-distinguishability can be observed using just a pair of qubits and a pair of incompatible measurements. We further show that violation of previously known notions of classicality such as preparation non-contextuality and Bell's local causality is a violation of ontic-distinguishability.
\end{abstract}

\maketitle

\section{Introduction.}
The operational picture of physical theories aims at understanding the physical world in terms of statistics generated by experiments performed on physical systems. The novelty of such a picture relies on the idea that observable facts are enough to predict outcomes of some physical experiments without much emphasis on the underlying physical states of reality,  commonly referred to as ontological states. However, there are limitations of the operational picture, not just towards predictive power but also that such a picture lacks any deeper understanding of nature. From a realist perspective, any operational prediction needs to be supported by some ontological states of reality even if they might be unobservable or hidden. One of the most compelling beliefs that are held by most physicists is that physical entities in the classical world are real or a fact that is independent of observation. However, in any general theory, it might also be possible that observable facts are just a reflection of some underlying physical reality even when the physical reality can not be directly observed. Building on this viewpoint, Bell in 1964 \cite{Bell} could prove that quantum theory is inherently different from classical physics.

One of the interesting avenues in the foundations of quantum theory is to understand how quantum theory diverges from classical theories. To facilitate such understanding, different notions of classicality have been suggested, like Bell's local causality \cite{Bell}, Kochen-Specker non-contextuality \cite{KS, Cabello}, and preparation non-contextuality \cite{PNC} to name a few. All such notions of classicality differ based on the assumptions imposed on the ontological level. If quantum theory violates any notion of classicality, then such a notion of classicality can also be understood as a no-go theorem for quantum theory, i.e., quantum theory can not be compatible with all the assumptions which define the notion of classicality. For example, Bell's theorem shows that any ontological description of quantum theory must violate the principle of local causality. Apart from their foundational importance, all such notions of classicality have led to various tasks in cryptography, communication and computation, which show a quantum advantage over classical strategies. For example, the security of the well-known cryptographic scheme $E91$ protocol is based on Bell's theorem \cite{E91}. In fact, every device-independent scheme is based upon violations of some notion of classicality \cite{di1, Yang, Bamps, sarkar, Ivan}.

In this work, we construct a novel notion of classicality, which is termed ontic-distinguishability, based on the idea that any physical entity in the classical everyday world is real and all of the entities can be perfectly distinguished from each other. The notion is based upon two assumptions, namely no-overlap and strong duality. Both of which, when taken together, reflect the fact that in classical physics, extremal states are states of reality and any sharp measurement perfectly observers the state of the system. Then, we construct an operational task that gives an upper bound for theories satisfying ontic-distinguishability. This bound can be violated in quantum theory using only two pure qubit states and two incompatible measurements, suggesting that even two-dimensional systems in quantum theory can not be described classically. also We further show that the notion of ontic-distinguishability implies all the other known ideas of classicality like preparation non-contextuality, Kochen-Specker non-contextuality and Bell's local causality. In discussions, we briefly provide some arguments on why ontic-distinguishability could be a more fundamental notion of classicality than contextuality.

\section{Preliminaries}
Before moving on to results, we would like to introduce the notations and relevant concepts required for this work.

\subsection{Operational picture}
Any experiment performed on some physical system can be understood as prepare, transform and measure experiment. In each run of the experiment, a preparation (the prepared system) undergoes a transformation (some dynamical process), after which the system is measured. 
Finally, the experiment is repeated a large number of times to gather enough statistics. In any such experiment, it is always assumed that each run of the experiment is statistically independent of other runs and the preparation, transformation, and measurement remain consistent throughout the runs of the experiment.  

Any preparation in the operational picture is denoted by $P$, transformation by $T$ and measurements are denoted by $M$ with outcomes denoted by $k$. The statistics are denoted by $p(k|M,P,T)$, which specifies the probability of obtaining the outcome $k$ of some measurement $M$ when the preparation $P$ undergoes a transformation $T$. For different physical theories, the probabilities are obtained in different ways. For example, in quantum theory, the preparations are given by density matrices $\rho$ belonging to some $d-$dimensional Hilbert space $\mathcal{H}_d$, the transformations are unitary matrices $U$, and the measurements are given by positive operator-valued measure (POVM) $\{M_k\}$.   The probabilities are obtained using the Born rule,
$p(k|M,P,T)=\Tr(U\rho U^{\dagger}M_k)$. For the rest of this manuscript, we would not consider systems that evolve. Consequently, the scenario can be simplified to only prepare and measure experiments where the probabilities are denoted by $p(k|M,P)$.

\subsection{Ontological models framework}
The idea of ontic states or hidden variables was put forward by Bell in his seminal work in 1964 \cite{Bell}. The idea was put more rigorously by Spekkens and Harrigan \cite{Spekkens2} and subsequently by Leifer \cite{Leifer}. The ontological models framework provides a basis for realist extensions of any physical theory. The ontic states can be understood as some "real physical states" that generally might not be observed directly. The ontic state is denoted by $\lambda$, which belongs to the ontological state space denoted by $\Lambda$. In general, any preparation procedure $P$ prepares a distribution of such ontic states denoted by $\mu(\lambda|P)$ which can be understood as the probability of preparing the ontic state $\lambda$ using a preparation procedure $P$. It is required that $\Lambda$ is a measurable space along with the condition that $\int_{\Lambda}\mu(\lambda|P)d\lambda=1$ for all preparations $P$ which signifies that every preparation must always prepare some state $\lambda\in\Lambda$ along with the constraint $\mu(\lambda|P)\geq 0$. The set $\Omega_{P}\in\Lambda$ represents the set of $\lambda's$ for some preparation $P$ such that the probability density $\mu(\lambda|P)\ne0$. Thus, we also have $\int_{\Omega_P}\mu(\lambda|P)d\lambda=1$.

Any measurement in the ontological models framework measures the ontic state. The measurements are represented as response functions denoted by, $\xi(k|M,\lambda)$ which specifies the probability of obtaining the outcome $k$ of some measurement $M$ given the ontic state $\lambda$. It is also required that the response function is measurable along with the condition $\xi(k|M,\lambda)\geq 0$ and $\sum_k\xi(k|M,\lambda)=1$ for all $\lambda\in \Lambda$.

Using the above prescription, any probabilities obtained in the prepare and measure experiments can be represented as,
\begin{eqnarray}
p(k|M,P)=\int_{\Lambda}\xi(k|M,\lambda)\mu(\lambda|P)d\lambda
\end{eqnarray}
Note that we assume that the preparation and measurement are related via the ontic state $\lambda$. This is known as $\lambda-$mediation.

\subsection{Ontic-Epistemic distinction}
One of the most intriguing questions in the foundations of quantum theory is whether a quantum state is a physical state of reality or epistemological knowledge of some underlying reality. 
To answer this question, the first requirement is to have a realist perspective and an ontological models framework that can describe quantum theory. If every pure quantum state corresponds to some unique set of ontological states, then the quantum state is ontic or corresponds to a real physical state of reality. On the other hand, if two pure quantum states share some ontological states, the quantum state is epistemic or just a representation of underlying physical states of reality. 
There has been a large number of ontological models, in some of which the quantum states are ontic like Bohmian mechanics \cite{Bohm1, Bohm2}, Beltrametti–Bugajski model \cite{BG} and Bell model \cite{BM} to name a few. Some other ontological models of quantum theory in which the quantum state is not a physical reality but just an epistemic knowledge are Kochen-Specker model \cite{KS}, Spekkens toy theory \cite{Spekkens3}, LJBR model \cite{LJBR} and ABCL models \cite{ABCL} to name a few.

\subsection{Extremal preparations and measurements}

Any preparation that can not be described as a convex mixture of two different preparations are known as extremal or pure preparations. Mathematically, if
\begin{eqnarray}
\overline{P}=z P'+(1-z)P''
\end{eqnarray}
then, $z=1$ or $P'=P''=\overline{P}$. Similarly, any extremal or sharp measurement $\overline{M}=\{e_i\}$, where $e_i$ represent the effects of the measurement, can not be realised as some convex combination of two different measurements $M'=\{e_i'\}$ and $M''=\{e_i''\}$.
Mathematically, if
\begin{eqnarray}
\overline{M}=z M'+(1-z)M''
\end{eqnarray}
that is,
\begin{eqnarray}
e_i=ze_i'+(1-z)e_i''\quad \forall i
\end{eqnarray}
then, $z=1$ or $M'=M''=\overline{M}$. Every preparation can be expressed as a classical mixture of extremal preparations and within classical and quantum theory any non-extremal preparation can also be realised as a extremal preparation in some higher dimensional system. Similarly, any measurement can be expressed as a convex mixture of extremal measurement and within classical and quantum theory any unsharp measurement can be realised as a sharp measurement acting on some higher dimensional system \cite{NielsenChuang}.
Now, we proceed towards the main result of this manuscript.
\section{Results}
We present here a novel measure of classicality termed ontic-distinguishability, which is based upon two physical assumptions well-supported by experiments to date. In this section, first, we introduce the assumptions defining ontic-distinguishability. Next, we construct a simple operational task and show that the success probability of the task using classical strategies is bounded from above. We find states and measurements within quantum theory, which gives a higher success probability than achievable using classical strategies. We begin by stating the assumptions.

\subsection{Assumptions}
The first assumption defining ontic-distinguishability is given as,
\begin{assu}[No overlap]\label{NO} For two extremal preparations $\overline{P}_1$ and $\overline{P}_2$, the corresponding ontic distributions $\mu(\lambda|\overline{P}_1)$ and  $\mu(\lambda|\overline{P}_2)$ do not overlap. Mathematically,
\begin{eqnarray}
\frac{1}{2}\int_\Lambda \left|\mu(\lambda|\overline{P}_1)-\mu(\lambda|\overline{P}_2)\right|d\lambda=1
\end{eqnarray}
\end{assu}
This assumption stems from the fact that extremal preparations generate ontological states that correspond to physically real states in any classical theory. As discussed above, two physically real states can not share the same ontological states. For simplicity, let us consider classical Hamiltonian dynamics. In classical Hamiltonian dynamics, any state that lies within the phase space is described by the position and momentum of the system $(x,p)$ that are physically measurable quantities and real. Going back to the ontological models framework, considering the ontic states to be the position and momentum and a preparation which prepares a state with some fixed energy $E$ (which is a distribution over the phase space states). It is known that energy is a physically measurable quantity. Consequently, if we can know the position and momentum of the system, we can predict the system's energy. However, if the same phase space state belonged to two different energy states, then it would not have been a physically measurable quantity. Note that classical mixture represents a lack of knowledge or ignorance that can be regarded as epistemic. For a remark, restricting to quantum theory the assumption of no-overlap is equivalent to psi-onticity or the quantum state is a physical state of reality. However, here we do not refer to quantum theory but to the fact that no-overlap is a feature of any classical theory. The second assumption involves measurements,

\begin{assu}[Strong duality]\label{SD}
For every extremal measurement $\overline{M}$ the response function belonging to the $k^{th}$ outcome, $\xi(k|\overline{M},\lambda)=1$ for every $\lambda\in\displaystyle\bigcup\limits_i \Omega_{\overline{P}_i}$ and $0$ for all other $\lambda\notin\displaystyle\bigcup\limits_i \Omega_{\overline{P}_i}$ such that $p(k|\overline{M},\overline{P}_i)=1$ where $\overline{P}_i$ are extremal preparations.

\end{assu}

The above assumptions stem from the fact that any sharp measurement in the classical world is just observing the system's state perfectly. In the classical world, measurements are not supposed to change the concerned system. Thus unlike quantum theory, postulates of any classical theory do not give an account of how one observes the system as it is a natural assumption that any physical state can in principle be observed without altering the state. For example, considering classical Hamiltonian dynamics, sharp measurements on the phase space just read the system's state $(x,p)$. In general, the measurements might not exactly suggest the state of the system but specify some region of phase space where the system resides. Such scenarios can be understood as measurements whose different outcomes are classically coarse-grained to a single outcome. Any unsharp measurement within classical theories is just a probabilistic mixture of sharp measurements. Consequently, without loss of generality within classical theories, it can always be assumed that measurements are sharp on some larger dimensional system. 

The term "strong-duality" comes from the fact that physical measurements strictly depend on physical preparations. A weaker assumption which is quite familiar in the quantum foundations community is known as outcome-determinism, which says that for any sharp measurement $M$, the response function for any outcome $\xi(k|\overline{M},\lambda)=1$ or $0$ for any $\lambda$. Strong-duality is much stricter in the sense that the response function is $1$ only over a specific region of ontic states defined by the extremal preparations. Now, we construct the operational task whose success probability is upper bounded for any classical strategy.

\subsection{Operational Task}
The task consists of two players, Alice and Bob, who are not allowed to communicate with each other during the run of the experiment. Alice has access to a preparation box that consists of two inputs $y=\{0,1\}$ which generate two different preparations $P_y$. Bob has access to a measurement box  with two inputs $x=\{0,1\}$ specifying two different measurements $M_x$ with any number of outputs strictly greater than 1, labelled as $a=\{0,1,2,\ldots\}$. Now, Alice and Bob choose their inputs independent of each other. As a consequence, a system is prepared by the preparation box with Alice, which is then sent to Bob, who measures the system as specified by the measurement box. The task is repeated a large number of times to collect statistics.

The quantifier of the probability of winning in the above-described task, also known as success probability $P_S$, is given by 

\begin{eqnarray}\label{sp}
P_S=\sum_{a,x,y}c_{a,x,y}p(a|x,y)
\end{eqnarray}
where $c_{a,x,y}=\frac{1}{4}$ if $a=x.y$ and $0$ otherwise.

Recalling from Bell's theorem, to observe a quantum advantage, it is required that the correlations shared between Alice and Bob are no-signalling \cite{Bellnon}. Similarly, to observe an advantage using quantum theory and constrain classical theories in the above described operational task \eqref{sp}, we are required to impose the condition that the measurements are of rank-one, which can be operationally defined as

\begin{figure}
    \centering
    \includegraphics[scale=.6]{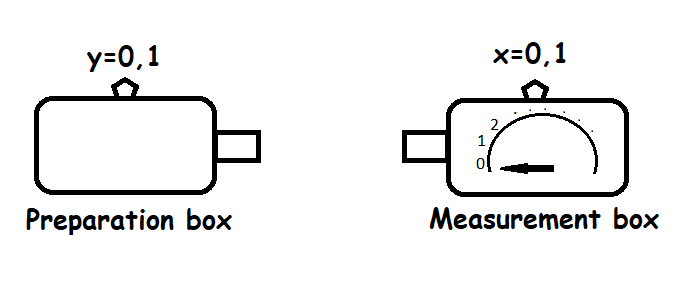}
    \caption{The operational task to observe quantum advantage. The preparation box consists of two inputs labelled by $y=0,1$ which generate two different preparations. The measurement box consists of two inputs labelled by $x=0,1$ which corresponds to two rank-one measurements of arbitrary number of outcomes labelled by $a=0,1,2,\ldots$}
    \label{fig1}
\end{figure}
\begin{udefn}[Rank-one measurements]\label{def3} For any measurement $M$, if two preparations $P$ and $P'$  give the same outcome with certainty, then $P$ and $P'$ can not be distinguished.

\end{udefn}
This means for any measurement $M$ and preparation $P$ and $P'$, if
\begin{eqnarray}
p(k|M,P)=p(k|M,P')=1\quad \text{for some $k$}
\end{eqnarray}
then $p(k|M,P)=p(k|M,P')$ for all outcomes $k$ of all measurements $M$. Note that, for sharp measurements if $p(k|M,P)=1$, then $P$ is an extremal preparation.

Constraining the strategies using the condition of measurements being rank-one def-\ref{def3}, we establish the following theorem on the success probability as defined in \eqref{sp}. Note that along with the above assumptions defining ontic-distinguishability, the implicit assumption of free will is required in any operational task. The assumption of free will ensures that Alice and Bob can freely choose respective inputs $(y,x)$ independent of each other's choices.
\begin{thm*}
For any theory which satisfies the assumptions of no overlap (assumption-\ref{NO}) and strong duality (assumption-\ref{SD}), the success probability \eqref{sp} is bounded from above as $P_S\leq \frac{3}{4}$.
\end{thm*}
\begin{proof}
Expanding the success probability \eqref{sp}, we have
\begin{eqnarray}\label{sp0}
P_S=\frac{1}{4} \left(p(0|0,0)+p(0|1,0)+p(0|0,1)+p(1|1,1)\right)
\end{eqnarray}
Since, the measurements consists of more than one outcome, we have $\sum_kp(k|1,1)=1$, which imposes that $p(0|1,1)+p(1|1,1)\leq 1$. Using this, we have

\begin{eqnarray}
P_S\leq\frac{1}{4} \left(p(0|0,0)+p(0|1,0)+p(0|0,1)-p(0|1,1)\right)+\frac{1}{4}\nonumber
\end{eqnarray}
Since there is no restriction on dimension, as discussed above any measurement can be realized as a sharp measurement on some larger system. Using this fact and the ontological models framework, we expand the above expression as
\begin{eqnarray}
p(0|0,0)+p(0|1,0)+p(0|0,1)-p(0|1,1)\nonumber\\=\int_\Lambda \xi(0|\overline{M}_0,\lambda)\left(\mu(\lambda|P_0)+\mu(\lambda|P_1)\right)d\lambda\nonumber\\+\int_\Lambda \xi(0|\overline{M}_1,\lambda)\left(\mu(\lambda|P_0)-\mu(\lambda|P_1)\right)d\lambda
\end{eqnarray}
Since the measurements are sharp, the assumption of strong duality (assumption-\ref{SD}) imposes that for all rank-one sharp measurements, the response function $\xi(0|\overline{M}_0,\lambda)=1$ for $\lambda\in\Omega_{\overline{P_0}}$ and $0$ for $\lambda\notin\Omega_{\overline{P_0}}$ and  $\xi(0|\overline{M}_1,\lambda)=1$ for $\lambda\in\Omega_{\overline{P_1}}$ and $0$ for $\lambda\notin\Omega_{\overline{P_1}}$ for some extremal preparations $\overline{P}_0$ and $\overline{P}_1$. This imposes that
\begin{eqnarray}
p(0|0,0)+p(0|1,0)+p(0|0,1)-p(0|1,1)\nonumber\\=\int_{\Omega_{\overline{P_0}}} (\mu(\lambda|P_0)+\mu(\lambda|P_1))d\lambda\nonumber\\+\int_{\Omega_{\overline{P_1}}} (\mu(\lambda|P_0)-\mu(\lambda|P_1))d\lambda
\end{eqnarray}
Now, the assumption of no-overlap (assumption-\ref{NO}) imposes that either ${\Omega_{\overline{P_0}}}$ and ${\Omega_{\overline{P_1}}}$ are disjoint for $\overline{P}_0$ and $\overline{P}_1$ being distinct or ${\Omega_{\overline{P_0}}}$ and ${\Omega_{\overline{P_1}}}$ are equivalent for $\overline{P}_0$ and $\overline{P}_1$ being same. Let's first consider the case when $\overline{P}_0$ and $\overline{P}_1$ are distinct.
\begin{eqnarray}\label{sp1}
p(0|0,0)+p(0|1,0)+p(0|0,1)-p(0|1,1)\nonumber\\=p_{0,0}+p_{0,1}+p_{1,0}-p_{1,1}\leq 2
\end{eqnarray}
where $p_{i,j}=\int_{\Omega_{\overline{P_i}}} \mu(\lambda|P_j)d\lambda$ and we used the fact that the ontological states generated from pure state preparations do not overlap, which imposes
\begin{eqnarray}
p_{0,0}+p_{1,0}&\leq& 1\quad \text{and},\nonumber\\
p_{0,1}+p_{1,1}&\leq&1
\end{eqnarray}
For the case when $\overline{P}_0$ and $\overline{P}_1$ being same, we have
\begin{eqnarray}\label{sp2}
p(0|0,0)+p(0|1,0)+p(0|0,1)-p(0|1,1)\nonumber\\=2p_{0,0}\leq 2
\end{eqnarray}
Thus, from \eqref{sp0}, \eqref{sp1} and \eqref{sp2} we can conclude that
\begin{eqnarray}
P_S\leq \frac{3}{4}.
\end{eqnarray}
This completes the proof.
\end{proof}
Since any classical theory satisfies the assumptions of no-overlap (assumption-\ref{NO}) and strong duality (assumption-\ref{SD}), the above theorem shows that using classical strategies, the maximum attainable value of the success probability is bounded from above by the value $P_s=\frac{3}{4}$. To better understand the above bound, let's consider a possible classical strategy. Let's say Alice sends her inputs $y=\{0,1\}$ to Bob. This can be realised by just sending a classical bit in which case, the constraint \ref{def3} imposes that $p(0|0,0)+p(0|0,1)\leq 1$ which in turn imposes from \eqref{sp} that $P_s\leq\frac{3}{4}$. Now, we show that there exist strategies in quantum theory which can attain a success probability $P_s=\frac{1}{2}+\frac{1}{2\sqrt{2}}$ which is strictly greater than $\frac{3}{4}$.

\subsection{Quantum advantage}
To witness quantum advantage, the preparations are chosen by Alice as, $P_0$ generates the state $\ket{\psi_0}=\ket{+}=\frac{1}{\sqrt{2}}\left(\ket{0}+\ket{1}\right)$ and $P_1$ generates the state $\ket{\psi_1}=\ket{0}$. The measurements are chosen by Bob represented as observables are given by, $\overline{M}_0=\frac{1}{\sqrt{2}}\left(\sigma_z+\sigma_x\right)$ and $\overline{M}_1=\frac{1}{\sqrt{2}}\left(\sigma_z-\sigma_x\right)$ where $\sigma_z$ and $\sigma_x$ corresponds to Pauli-$z$ and Pauli-$x$ matrices respectively. Calculating the necessary probabilities  to evaluate the success probability \eqref{sp} using the above chosen quantum states and measurements, we have
\begin{eqnarray}
p(0|0,0)=p(0|1,0)=p(0|0,1)=p(1|1,1)\nonumber\\=\cos^2\left(\frac{\pi}{8}\right)
\end{eqnarray}
Thus, the success probability is given by
\begin{eqnarray}
P_s=\cos^2\left(\frac{\pi}{8}\right)=\frac{1}{2}+\frac{1}{2\sqrt{2}}.
\end{eqnarray}
The existence of such strategies in quantum theory shows that quantum theory outperforms classical strategies in the above presented operational task. Also, such a violation shows that the assumptions of no-overlap (assumption-\ref{NO}) and strong duality (assumption-\ref{SD}) are not compatible with quantum theory. Thus, it is impossible to construct any such ontological model of quantum theory which have the features of no-overlap (assumption-\ref{NO}) and strong duality (assumption-\ref{SD}) together. However, some ontological models of quantum theory can exist that violate either of the two assumptions. Some notable ontological models of quantum theory that violate the assumption of strong duality (assumption-\ref{SD}) but satisfy the assumption of no-overlap (assumption-\ref{NO}) are  Beltrametti–Bugajski model \cite{BG} and Bell model \cite{BM}. Another class of ontological models like the Kochen-Specker model \cite{KS} and Spekkens toy theory \cite{Spekkens3} violate the assumption of no-overlap (assumption-\ref{NO}) but satisfy the assumption of strong duality (assumption-\ref{SD}). In the next section, we show how the notion of ontic-distinguishability implies some of the well-known notions of classicality.

\section{Relations to other notions of classicality}
In this section, we show that ontic-distinguishability implies the previously known notions of classicality. We need an additional assumption concerning non-extremal preparations.

\begin{defn}[Convexity]\label{convex} For any preparation $P$ that can be realised as a classical mixture of some preparations $P_i's$ such that $P=\sum_iz_iP_i$, then the corresponding ontological distribution is given by $\mu(\lambda|P)=\sum_iz_i\mu(\lambda|P_i)$ where $z_i\geq0$ and $\sum_iz_i=1$.
\end{defn}
The above definition holds true for any preparations $P_i$. However, for this work, we only consider $P_i's$ which are pure states, as any preparation which generates a mixed state can be understood as some classical mixture of pure states. As pure states represent physical states of reality in classical theories, any mixture of such pure states would be equivalently reinstated at the ontological level.

First, we prove that ontic-distinguishability and convexity implies preparation non-contextuality \cite{PNC}.
For this, let us first recall the definition of preparation non-contextuality. Preparation non-contextuality imposes that if two preparations are indistinguishable, then their respective ontological distributions are same. Mathematically, if
\begin{eqnarray}\label{PNC1}
p(k|M,P_1)=p(k|M,P_2)\quad \forall k,M
\end{eqnarray}
then, $\mu(\lambda|P_1)=\mu(\lambda|P_2)$ for all $\lambda\in \Lambda$.

Now, we would show that the assumptions of no overlap (assumption-\ref{NO}) and strong duality (assumption-\ref{SD}) along with the assumption of convexity def-\ref{convex} when imposed on the ontological state space imposes preparation non-contextuality. For this, let us note that as discussed before any measurement can be realized as a sharp measurement in some higher dimensional space. Expressing \eqref{PNC1} in the ontological models framework,
\begin{eqnarray}
\int_\Lambda \xi(k|\overline{M},\lambda)\mu(\lambda|P_1)d\lambda=\int_\Lambda \xi(k|\overline{M},\lambda)\mu(\lambda|P_2)d\lambda
\end{eqnarray}
for all outcomes $k$ of all sharp measurements $\overline{M}$. To show $\mu(\lambda|P_1)=\mu(\lambda|P_2)$, it is enough to consider rank-one measurements def-\ref{def3}. As discussed before, the assumption of strong duality (assumption-\ref{SD}) imposes that $\xi(k|\overline{M},\lambda)=1$ for $\forall\lambda\in\Omega_{\overline{P}}$ and $0$ for all other $\lambda\notin\Omega_{\overline{P}}$ for some pure state preparation $\overline{P}$,
\begin{eqnarray}
\int_{\Omega_{\overline{P}}} \mu(\lambda|P_1)d\lambda=\int_{\Omega_{\overline{P}}}\mu(\lambda|P_2)d\lambda
\end{eqnarray}
Now, assuming convexity def-\ref{convex} we have
\begin{eqnarray}\label{PNC2}
\sum_iz_{i,1}\int_{\Omega_{\overline{P}}} \mu(\lambda|\overline{P}_{i,1})d\lambda=\sum_iz_{i,2}\int_{\Omega_{\overline{P}}} \mu(\lambda|\overline{P}_{i,2})d\lambda
\end{eqnarray}
Assuming no overlap (assumption-\ref{NO}), it can be concluded that the above equality \eqref{PNC2} holds iff  $\overline{P}_{i,1}=\overline{P}_{i,2}=\overline{P}$ and $z_{i,1}=z_{i,2}$ for some $i$. Using a similar argument for all outcomes $k$ of all rank-one sharp measurements $\overline{M}$ we can conclude that $\mu(\lambda|P_1)=\mu(\lambda|P_2)$ for all $\lambda\in \Lambda$.

Now, we show that ontic-distinguishability and convexity also gives rise to a recently proposed notion of classicality termed bounded ontological distinctness \cite{BOD}. For this, let us first recall the definition of bounded ontological distinctness $(BOD_P)$.  Bounded ontological distinctness imposes that for distinguishability of preparations at the ontological level is same as distinguishability at the operational level. Mathematically, for n-preparations $P_x$ when $x\in{0,1,\ldots,n-1}$ if
\begin{eqnarray}\label{BOD1}
\max_{M}\frac{1}{n}\sum_xp(k=x|M,P_x)=p
\end{eqnarray}
then,
\begin{eqnarray}\label{BOD2}
\frac{1}{n}\int_{\Lambda}\max_x\left\{\mu(\lambda|P_x)\right\}d\lambda=p
\end{eqnarray}

Now, we would show that the assumptions of no overlap (assumption-\ref{NO}) and strong duality (assumption-\ref{SD}) along with the assumption of convexity def-\ref{convex} when imposed on the ontological state space imposes bounded ontological distinctness. Expressing \eqref{BOD1} in the ontological models framework,
\begin{eqnarray}
\max_{\overline{M}}\frac{1}{n}\sum_x\int_{\Lambda}\xi(k=x|\overline{M},\lambda)\mu(\lambda|P_x)=p
\end{eqnarray}
Assuming convexity def-\ref{convex}, any ontological distribution can be written as $\mu(\lambda|P_x)=\sum_{i}z_{i,x}\mu(\lambda|\overline{P}_i)$, we have
\begin{eqnarray}
\max_{\overline{M}}\frac{1}{n}\sum_x\sum_{i}z_{i,x}\int_{\Lambda}\xi(k=x|\overline{M},\lambda)\mu(\lambda|\overline{P}_{i})=p
\end{eqnarray}
Assuming strong duality (assumption-\ref{SD}) and no overlap (assumption-\ref{NO}) any measurement $\overline{M}$ that maximises the distinguishing probability must have the response functions $\xi(k=x|\overline{M},\lambda)$ such that they pick up the probabilities $z_{i,x}$ which is maximised for each $x$. Thus, we have
\begin{eqnarray}
\frac{1}{n}\sum_i\max_xz_{i,x}=p
\end{eqnarray}

Now, evaluating \eqref{BOD2} assuming convexity def-\ref{convex} we have,
\begin{eqnarray}
\frac{1}{n}\int_{\Lambda}\max_x\left\{\sum_iz_{i,x}\mu(\lambda|\overline{P}_{i})\right\}d\lambda\nonumber\\=\frac{1}{n}\sum_i\max_xz_{i,x}=p
\end{eqnarray}
where we arrived at the above expression by assuming that pure states do not overlap (assumption-\ref{NO}).

Using some well established results as stated in \cite{BOD} among different notions of classicality we have the following hierarchy,
\begin{eqnarray}
\text{ontic-distinguishability}&\implies& BOD_{P}\nonumber\\&\implies& \text{Prep non-contextuality}\nonumber\\&\substack{QT\\\implies}&\text{K-S non-contextuality}\nonumber\\&\substack{NS\\\implies}&  \text{Bell local-causality}
\end{eqnarray}

Thus assuming that convexity def-\ref{convex} holds true, any violation of Bell's local causality, Kochen-Specker non-contextuality, preparation non-contextuality or bounded ontological distinctness imposes violation of ontic-distinguishability. The assumption of convexity is natural for any classical theory; even for quantum theory, there is no successful ontological model that violates convexity.  

\section{Discussions}
In the presented work, from a realist perspective, we first constructed a novel notion of classicality termed ontic-distinguishability based on just two assumptions imposed on the ontological state space. We justified both the assumptions based on physical arguments, which are satisfied for any known physical theories of the classical world. Then, we constructed an operational task bounded for any theories that satisfy the assumptions of ontic-distinguishability. Kochen-Specker non-contextuality and Bell's local causality requires systems of atleast dimension three and four respectively to show a quantum advantage. Bounded ontological distinctness and preparation non-contextuality can be violated using only qubits; however, they require atleast three and four qubit preparations respectively to show a quantum advantage. Contrary to this, ontic-distinguishability can be violated using just a pair of qubits, making the above presented operational task much more applicable in practical scenarios. Then, we showed that the notion of ontic-distinguishability implies all the other known notions of classicality like Bell's local causality.

As argued by Spekkens \cite{Spekk2}, any good notion of classicality must have three properties; first, it can be experimentally tested, second, it constitutes a resource and third that it applies to a wide range of scenarios. The operational task provided in this work is a way to test the notion of ontic-distinguishability experimentally. Further, it applies to scenarios where minimal resources are required in terms of the number of preparations and measurements than previous notions of classicality which makes it applicable to a larger number of scenarios. Although any quantum advantage constitutes a resource, further works are needed to know the extent where violation of ontic-distinguishability serves as a resource.

In recent years, a lot of work has been put forward to establish generalised non-contextuality \cite{Spekk3} as the current best notion of classicality. Generalised non-contextuality is based upon the idea of Leibniz principle of identity of indiscernibles; that is, if two entities are equivalent at the operational level, then they are equivalent at the ontological level too. Any classical theory should indeed be non-contextual; however, it is still not clear whether every non-contextual theory should be considered classical as is argued in \cite{Barrett} that "a classical theory is one for which all systems are classical" or equivalently that all of the underlying systems can be explained using classical laws of nature. It could well be that some non-contextual theory exists that any known classical laws of nature can not explain. On the other hand, the notion of ontic-distinguishability is based on one of the basic ideas of the classical everyday world that physical entities are real and every sharp measurement observes the state of the system perfectly. It is not natural that our classical worldview is considered to be epistemic when generations of classical physicists from Galileo to Einstein have argued that facts of the world are real and independent of our knowledge of the world as Einstein says, "Do you really believe the moon is not there when you are not looking at it?". Quantum theory challenges this worldview, and we need to give up either the fact that physical states are real ((assumption-\ref{NO})) or quantum measurements are strongly dual ((assumption-\ref{SD})). Some non-contextual models, such as Spekkens toy theory, violate the assumption of no-overlap and thus, according to the notion of ontic-distinguishability, would not be considered classical. Bell's model, on the other hand, considers that physical states are real but give up the fact that quantum measurements are strongly dual. Some GPT's such as generalised local theory (GLT) and generalised no-signalling theory (GNST), are not classical with respect to ontic-distinguishability. Some fragments of quantum theory such as Gaussian quantum mechanics or linear quantum optics that have been shown to have epistemic ontological models \cite{Spekk4} which again should not be considered classical. However, with respect to ontic-distinguishability, some quantum phenomena such as interference and tunnelling which have analogues in classical electromagnetic theory are classical. 

Interestingly, when the assumptions of no-overlap and strong duality are imposed on the ontology, in the resulting generalised probabilistic theory (GPT), the state space is simplex, and the effect space is the dual of this simplex. This has been recently termed as simpliciality \cite{Spekk3}. This further strengthens the presented notion as suggested in \cite{Barrett} that a GPT should be considered classical if and only if the state space is simplex and the measurement effects are its dual. However, one of the recent works \cite{Spekk3} argue that a GPT should be deemed classical if the state space is simplex-embeddable and the effect space as its dual. However, they also prove that simplex-embeddability is equivalent to simpliciality if the GPT's satisfy no-restriction hypothesis \cite{Chiri}, that is, all states and effects that are permitted in the GPT are physically realisable. In the present work, we are concerned with only preparations and measurements that are physically realisable, and thus we expect that the GPT should be simplicial. It might always be possible that the corresponding GPT could be extended by adding more states and effects such that it would be simplex-embeddable and not simplicial, for instance, adding extra states to the GPT that can not be realised in any possible preparation as the assumption of no-overlap is concerned with the physical preparations only.

Several follow up questions arises from this work. From a quantum foundations perspective, ontic-distinguishability might give better lower bounds to the overlap of the ontological states of quantum theory than previously known bounds \cite{QO1,QO2}. Also, the above theorem serves as a no-go theorem for quantum theory. Further, the notion of ontic-distinguishability would give a better understanding of the set of quantum correlations. As was shown above, the violation of known notions of classicality would also imply a violation of ontic-distinguishability. As a consequence, the set of quantum correlations is enlarged, and this might give interesting results towards boundedness of quantum correlations \cite{Cola}. In this work, it is shown that ontic-distinguishability implies preparation non-contextuality; it would also be interesting to see the relation between ontic-distinguishability and generalised non-contextuality, which is infact non-contextuality for preparations, transformations as well as measurements. From quantum information perspective, the above presented operational task might give rise to cryptographic and communication tasks which need less resource to execute as entanglement is not needed to show a quantum advantage. Also, some device-independent schemes might be possible to construct employing the idea of the presented operational task.
\\
\begin{center}
    \textbf{Acknowledgements}
\end{center}

We would like to thank Remigiusz Augusiak and Debashish Saha for fruitful discussions. 
\\

\textbf{Funding} This work is supported by Foundation for Polish Science through the First Team project (No First TEAM/2017-4/31).
\\
\\

\textbf{Data Availability} There are no associated data for this paper.
\begin{center}
    \textbf{Declarations}
\end{center}

{\bf{Conflict of interest}} The author declares no conflict of interest.

\end{document}